\documentclass[runningheads]{llncs}
\usepackage[T1]{fontenc}
\usepackage[utf8]{inputenc}
\usepackage{amsmath,amssymb,amsfonts}
\usepackage{bm}
\usepackage{graphicx}
\usepackage{tikz}
\usepackage{pgfplots}
\pgfplotsset{compat=1.15}
\usepackage{hyperref}
\usepackage{booktabs}
\usepackage{array}

\usepackage{enumerate}
\usepackage{enumitem}
\usepackage{microtype}
\usepackage{xcolor}

\title{Cost of Manipulation in AMM-Based Oracles}
\author{Sebastian M\"uller\inst{1,2} \and
  Nordine Moumeni \inst{1} \and
  Adel Messaoudi\inst{1}}
\institute{Aix Marseille Univ, CNRS, I2M, Marseille, France \and
  IOTA Foundation, Germany}
\date{}

\begin{document}
\maketitle
\begin{abstract}
We study the robustness of AMM-based on-chain price oracles to 
strategic manipulation. An attacker trades against constant product
automated market makers (CPMMs) to distort an on-chain oracle, 
arbitrageurs restore cross-pool and cross-venue consistency, and an oracle designer chooses how to aggregate pool quotes. 
  
Taking an efficient-market-hypothesis (EMH) view of the off-chain 
``true'' price, we define the \emph{cost of manipulation} as the minimal mark-to-market loss that an attacker must incur to move the oracle by
a given multiplicative factor. For independent CPMMs, we derive 
	closed-form single-pool manipulation formulas and solve the 
	attacker–designer game for weighted means and weighted medians, 
	showing that liquidity weights maximize the minimum cost
	of manipulation within these classes for weighted medians (for any distortion level) and, for weighted means, locally as the distortion tends to zero. For larger distortions, weighted means become more fragile: optimal weights can depend on the target distortion and no single choice is uniformly optimal across distortion levels. 
	In a frictionless CPMM model with cross-pool arbitrage, 
	the manipulation cost depends only on the total quote depth
	and coincides across symmetric aggregators. 

We extend this framework to multi-asset star architectures, confirming that liquidity weights remain optimal in the same sense. Finally, we bridge theory and practice by incorporating dwell times and rate limits, providing a quantitative yardstick to size oracles against the explicit economic costs of attack.
\end{abstract}

\section{Introduction}

Automated market makers (AMMs) have become core infrastructure in decentralized finance (DeFi), serving as venues for liquidity provision and as a source of on-chain price information. Lending protocols, derivatives platforms, and stablecoins routinely rely on on-chain oracles to trigger liquidations, margin calls, and redemptions, often for large notional exposures. In current practice, most production oracles still aggregate off-chain data from centralized exchanges, APIs, or market makers, but a smaller and conceptually important class of ``pure on-chain'' oracles derives prices solely from AMM activity. These designs improve decentralization and verifiability by avoiding external data feeds, yet their security and robustness against manipulation remain only partially understood.

We develop a model of the \emph{cost of manipulation} for such AMM-based oracles. An \emph{attacker} strategically trades against constant-product market makers (CPMMs) in order to push the oracle away from a latent efficient price, an \emph{oracle designer} chooses how to aggregate pool quotes, and \emph{arbitrageurs} trade across pools and external venues to restore price consistency. Adopting an efficient-market-hypothesis (EMH) view of the off-chain price, we measure manipulation cost as the mark-to-market economic loss that an attacker must realize in order to move the oracle by a prescribed multiplicative factor \(r\ge 1\), upward or downward. Our goal is to provide simple, closed-form benchmarks for the robustness of pure on-chain CPMM-based oracles that can serve as a foundation for more elaborate oracle designs.

	At a high level, our results show that in CPMMs the cost of manipulating an oracle admits simple expressions that depend only on liquidity depth and aggregation weights, and that within natural mean- and median-type aggregators, liquidity-proportional weights are max--min optimal (exactly for weighted medians and to second order in a small-distortion expansion for weighted means). The rest of the paper makes this statement precise in single-pool, multi-pool, and multi-asset settings and quantifies the resulting robustness benchmarks for AMM-based oracles.
	
	\paragraph{Mean vs.\ median across distortion levels.}
	For small distortions \(r=1+\varepsilon\), liquidity-weighted means are more expensive to manipulate than liquidity-weighted medians at leading order, reflecting that means ``use'' all pools while medians require only a majority-weight cover.
	For larger distortions, however, weighted means can become substantially more fragile because per-pool CPMM manipulation cost grows sublinearly once a pool is pushed beyond the inflection threshold of the single-pool factor. In this regime, designer-optimal mean weights can depend on the target distortion and there is generally no distortion-uniform max--min choice.
\paragraph{Contributions and Organization.}
Formally, the paper makes five contributions:
\begin{itemize}[leftmargin=*]
	  \item We recall closed-form single-pool formulas for trade
	   size and economic loss as functions of the relative 
	   distortion factor \(r\) under a CPMM, and 
	   extend them to include proportional swap fees via
	  a simple rescaling of inputs (Section~\ref{sec:cpmm-primer}, Appendix~\ref{app:single-pool-derivations}, and Section~\ref{sec:fees}).
  \item We characterize the optimal attacker strategies
   for multi-pool oracles based on weighted means and weighted medians and solve the corresponding designer’s problem. Liquidity weights are max--min optimal for weighted medians for any distortion level, and are locally max--min optimal for weighted means as \(r\to 1\); we also provide an explicit counterexample showing that weighted means admit no distortion-uniform max--min choice beyond small distortions (Section~\ref{sec:single-pair}, Section~\ref{sec:npool-exact}, Theorems~\ref{thm:weighted-mean-opt} and~\ref{thm:median-opt-weights}, and Appendices~\ref{app:weighted-mean-proof}--\ref{app:weighted-median-proof}).
  \item We analyze frictionless cross-pool arbitrage
   and swap fees, noting that under perfect equalization any symmetric aggregator yields the same oracle price and that manipulation cost depends only on total depth, while fees induce no-arbitrage bands that jointly affect arbitrage efficiency and attack cost (Sections~\ref{sec:robustness-metric},~\ref{sec:single-pair}, and~\ref{sec:fees}).
  \item We sketch a multi-asset extension for star architectures around a numéraire and show that liquidity weights remain optimal asset-by-asset (Section~\ref{sec:multi-asset} and Appendix~\ref{app:multi-asset-star}).
  \item We discuss practical extensions that embed our static cost benchmarks into dynamic and implementation-aware settings, introducing a dwell-time metric \(\mathrm{Cost}_A(r,\tau)\) and outlining how platform-specific rate limits and gas costs interact with manipulation incentives in practice (Section~\ref{sec:discussion}).
\end{itemize}

\paragraph{Relation to prior work.}
Our analysis lies at the intersection of CFMM microstructure, oracle manipulation in DeFi, and robust aggregation.
On the CFMM side, Uniswap v2/v3 and the CFMM literature~\cite{uniswapv2,uniswapv3,angeris2020improved} characterize invariant-based pricing, slippage, and arbitrage but do not formulate a general cost of manipulation metric or a max–min designer problem over cross-pool weights.
Oracle-security work~\cite{zhang2020firstlook,qin2021attacks,yang2021flashloan,aspembitova2022oracles,adams2022uniswap,chaos2023twap} studies concrete mechanisms such as TWAPs and documents manipulation events and capital requirements, while treating AMM pricing largely as a black box.
Robust-aggregation results~\cite{hampel2001robust,allouche2024crypto,six2020cryptoindices} motivate volume- and liquidity-weighted means and medians statistically but abstract away from bonding-curve mechanics and strategic attackers.
Our contribution is to connect these strands by deriving closed-form CPMM manipulation costs, solving the associated attacker–designer games for weighted means and medians (in single-asset and star-architecture settings), and showing that liquidity-proportional weights are max–min optimal for weighted medians (for any distortion level) and locally max--min optimal for weighted means near \(r=1\), while weighted means admit no distortion-uniform max--min weights in general. We refer to Section~\ref{sec:related-work} for a more detailed discussion of related literature.

\section{Background and Model}\label{sec:background}

\subsection{Primer on CPMMs and Notation}\label{sec:cpmm-primer}

This subsection collects standard definitions and single-pool formulas for constant-product market makers and fixes the notation used throughout the paper.
In this paper we focus exclusively on Uniswap v2–style \emph{constant-product} AMMs, which we refer to as \emph{constant-product market makers (CPMMs)}: pools whose reserves \((x,y)\) satisfy the invariant \(xy=k\) \cite{uniswapv2,angeris2020improved}. 

\paragraph{Constant-product invariant and marginal price.}
Consider a trading pair \((X,Y)\) with reserves \((x,y)\) in a CPMM pool. The invariant is
\[
  k = x\,y,\qquad k>0,
\]
so any trade moves \((x,y)\) along the curve \(xy=k\). The marginal price of \(X\) in units of \(Y\) is the rate at which the reserves trade locally:
\[
  \text{marginal price of }X\text{ in }Y \;=\; -\frac{dy}{dx} = \frac{y}{x}.
\]
Thus the reserve ratio
\[
  p := \frac{y}{x} \quad (Y\text{ per }X)
\]
is exactly the marginal price on the CPMM curve. This is why on-chain ``spot prices'' for CPMMs are typically read as \(p=y/x\).

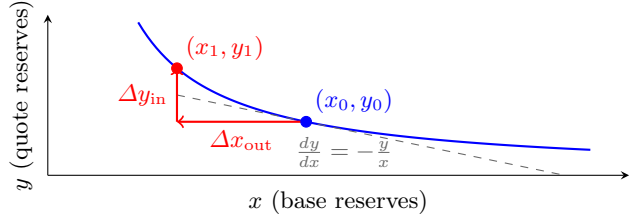
\begin{figure}[t]
  \centering
  \begin{tikzpicture}
    \begin{axis}[
      width=0.76\textwidth,
      height=3.7cm,
      axis lines=left,
      xlabel={$x$ (base reserves)},
      ylabel={$y$ (quote reserves)},
      xtick=\empty, ytick=\empty,
      xmin=0, xmax=4.5,
      ymin=0, ymax=4.5,
      grid=major
    ]
      \addplot[domain=0.7:4.2, samples=100, thick, blue] {3/x};

      \addplot[domain=1:4, dashed, black!60] {-0.75*x + 3};
      \node[black!60, anchor=south west] at (axis cs:1.85,.0) {$\frac{dy}{dx}=-\frac{y}{x}$};

      \addplot[mark=*, blue] coordinates {(2,1.5)};
      \node[above right, blue] at (axis cs:2,1.5) {$(x_0,y_0)$};
      \addplot[mark=*, red] coordinates {(1,3)};
      \node[above right, red] at (axis cs:1,3) {$(x_1,y_1)$};

      \draw[->, thick, red] (axis cs:2,1.5) -- (axis cs:1,1.5) node[midway, below] {$\Delta x_{\mathrm{out}}$};
      \draw[->, thick, red] (axis cs:1,1.5) -- (axis cs:1,3) node[midway, left] {$\Delta y_{\mathrm{in}}$};
    \end{axis}
  \end{tikzpicture}
  \caption{The CPMM bonding curve \(xy=k\) and a \(Y\!\to\!X\) trade that increases the marginal price \(p=y/x\).}
  \label{fig:bonding_curve}
\end{figure}

\paragraph{Liquidity depth and the role of \(k\).}
For a fixed external (efficient) price \(p^{\star}\), the invariant \(k\) encodes the pool’s \emph{liquidity depth}. At equilibrium \(p=p^{\star}\) we have
\[
  x_0 = \sqrt{\frac{k}{p^{\star}}},\qquad y_0 = \sqrt{k\,p^{\star}},
\]
so both reserves scale as \(\sqrt{k}\). Moving the on-chain price from \(p_0\) to a new level \(p_1\) requires a finite trade that shifts reserves along the same curve. We parametrize distortions multiplicatively: for a factor \(r\ge 1\), an ``upward'' move targets \(p_1=r\,p_0\) and a ``downward'' move targets \(p_1=p_0/r\). This yields a symmetric notion of manipulation, since the relevant cost factor satisfies \(f(r)=\sqrt r+1/\sqrt r-2=f(1/r)\) as shown below.

\paragraph{Baseline single-pool formulas (as functions of a multiplicative distortion \(r\ge 1\)).}
Consider a CPMM initially at price \(p_0\) with reserves \((x_0,y_0)\) and invariant \(k=x_0y_0\).
\begin{itemize}[leftmargin=*]
  \item \emph{Upward move (factor \(r\)).} To move the on-chain price to \(p_1=r\,p_0\) with a \(Y\!\to\!X\) trade, the attacker must supply
  \begin{equation}\label{eq:y_in_single}
    y^{\mathrm{in}}_{\uparrow}(r) \;=\; y_0\!\left(\sqrt{r} - 1\right).
  \end{equation}
  \item \emph{Downward move (factor \(1/r\)).} To move the price to \(p_1=p_0/r\) with an \(X\!\to\!Y\) trade, the attacker must supply
  \begin{equation}\label{eq:x_in_single}
    x^{\mathrm{in}}_{\downarrow}(r) \;=\; x_0\!\left(\sqrt{r} - 1\right).
  \end{equation}
  \item \emph{Economic cost (mark-to-market at \(p^{\star}=p_0\)).} When the reference price equals the pre-trade pool price, the economic loss in units of \(Y\) is
  \begin{equation}\label{eq:upward_cost}
    C_Y^{\uparrow}(r) \;=\; y_0\!\left( \sqrt{r} + \frac{1}{\sqrt{r}} - 2 \right),
  \end{equation}
  \begin{equation}\label{eq:downward_cost}
    C_Y^{\downarrow}(r) \;=\; y_0\!\left( \sqrt{r} + \frac{1}{\sqrt{r}} - 2 \right).
  \end{equation}
\end{itemize}
The trade sizes \eqref{eq:y_in_single}–\eqref{eq:x_in_single} 
reflect the \(\sqrt{k}\)-scaling of liquidity depth, 
and the cost formulas \eqref{eq:upward_cost}–\eqref{eq:downward_cost}
encode the realized slippage loss relative to the efficient price. 
Detailed derivations are recalled in Appendix~\ref{app:single-pool-derivations}.
These CPMM price-impact and slippage expressions coincide with the 
formulas used in protocol documentation and CFMM analyses; see, for example, 
the Uniswap v2 and v3 whitepapers \cite{uniswapv2,uniswapv3} and the
 CFMM framework of Angeris and Chitra \cite{angeris2020improved}.

\paragraph{Notation and reference price.}
Throughout, we use:
\begin{itemize}[leftmargin=*]
  \item \((X,Y)\): base and quote assets.
  \item \((x,y)\): current pool reserves; \(k=xy\) the CPMM invariant.
  \item \(p = y/x\): on-chain marginal price (\(Y\) per unit \(X\)).
  \item \(p^{\star}\): external ``fair'' or efficient price, unaffected by on-chain trades (e.g., a robust cross-venue midprice).
\end{itemize}
When we consider several pools for the same pair, we write \((x_i,y_i)\), \(k_i\), and \(p_i=y_i/x_i\) for pool \(i\), and collect pool prices into a vector \(\bm{p}=(p_i)_i\). An oracle or price aggregator applies a deterministic functional \(A\) (e.g., liquidity-weighted mean, weighted median, trimmed mean) to \(\bm p\) to obtain the reported oracle price \(\hat p = A(\bm p)\).

\paragraph{Agents and manipulation cost.}
The strategic agents in our model are:
\begin{itemize}[leftmargin=*]
  \item an \emph{oracle designer}, who chooses the aggregation rule \(A\) and fee/parameter settings;
  \item an \emph{attacker}, who submits trades to CPMM pools with the goal of pushing \(\hat p\) away from \(p^{\star}\);
  \item \emph{arbitrageurs}, who trade across pools (and vs.\ external venues) whenever price discrepancies exceed their frictions.
  \end{itemize}
We measure manipulation cost in units of the quote asset \(Y\) as the mark-to-market loss of the attacker’s trades when valued at \(p^{\star}\) (slippage loss relative to the efficient price). This quantity underlies the ``adversarial cost'' or ``manipulation cost'' studied in the rest of the paper.

\subsection{Robustness Metric: Cost of Manipulation}\label{sec:robustness-metric}

We formalize the notion of manipulation cost used throughout the paper. Let \(p^{\star}\) denote the efficient price of a given asset in units of a reference asset, and let \(A\) be an aggregation rule mapping observed pool quotes (and possibly other on-chain data) to an oracle output \(\hat p\). We consider a one-shot setting first and then indicate how to extend to a multi-step dwell requirement.

\begin{definition}[Cost of manipulation]
Fix an aggregation rule \(A\), an efficient price \(p^{\star}\), a distortion factor \(r\ge 1\), and a reference asset (e.g., the quote asset \(Y\) or a numéraire). Let \(\mathcal{U}\) denote the set of admissible attack strategies (collections of on-chain trades on the relevant pools), and write \(C(u)\) for the mark-to-market economic loss of an attack \(u\in\mathcal{U}\) measured in the reference asset at the efficient prices. The cost of manipulation at level \(r\) is
\[
  \mathrm{Cost}_A(r)
  := \inf\Big\{ C(u) : u\in\mathcal{U},\ \max\Big\{\frac{\hat p(u)}{p^{\star}},\,\frac{p^{\star}}{\hat p(u)}\Big\} \ge r\Big\},
\]
where \(\hat p(u) := A(\text{quotes after }u)\) is the oracle output after the attack has been executed and before any corrective arbitrage.
\end{definition}

In words, \(\mathrm{Cost}_A(r)\) is the minimal economic loss an attacker must incur to move the oracle by a factor of at least \(r\) up or down relative to \(p^{\star}\). In the single-CPMM setting with \(A\) equal to the pool price, \(\mathrm{Cost}_A(r)\) reduces to the one-pool formulas \(C_Y^{\uparrow,\downarrow}(r)\) in \eqref{eq:upward_cost}–\eqref{eq:downward_cost}. In the multi-pool mean and median settings, \(\mathrm{Cost}_A(r)\) is given by the solutions of the corresponding multi-pool optimization problems. In the rest of the paper we focus on this ``one-step'' notion under static efficient prices; a dynamic extension with dwell-time constraints is discussed in the outlook section.

\section{Manipulation of Single-Pair Oracles}\label{sec:single-pair}

Having fixed notation and the robustness metric, we now analyze how much economic loss an attacker must incur to manipulate oracles built from one or several CPMM pools that all trade the same asset pair. We start with two pools, then consider the general case of $N$ independent pools aggregated by means or medians, and finally incorporate cross-pool arbitrage.

We will repeatedly use the single-pool cost factor
\[
  f(r) := \sqrt r + \frac{1}{\sqrt r} - 2,
\]
so that, when \(p^{\star}=p_0\), moving a CPMM price from \(p_0\) to \(r\,p_0\) (or to \(p_0/r\)) costs \(C_Y(r)=y_0 f(r)\) in quote units; see \eqref{eq:y_in_single}–\eqref{eq:downward_cost}.

\subsection{Two CPMM Pools}
As a warm-up and to illustrate the main ideas for the general \(N\)-pool setting, we first analyze the case of two independent 
CPMM pools aggregated by a deterministic oracle \(A\).
We index pools by a subscript and indicate time by a superscript in parentheses.
Initial (pre-attack) reserves and prices are \((x_i^{(0)},y_i^{(0)})\) with invariant \(k_i=x_i^{(0)}y_i^{(0)}\) and price \(p_i^{(0)}=y_i^{(0)}/x_i^{(0)}\). 
An attacker produces post-trade quotes \(p_i^{(1)}\). 
The oracle aggregates the updated \(\{p_i^{(1)}\}\) with one of the
following methods.

\subsubsection{Weighted Mean}
Let weights \(w_1,w_2>0\) with \(w_1+w_2=1\) (e.g., by liquidity depth \(w_i\propto\sqrt{k_i}\), equivalently \(w_i\propto y_i^{(0)}\) when pre-trade prices match). The aggregated post-trade price is
\[ \hat p^{(1)} = w_1\, p_{1}^{(1)} + w_2\, p_{2}^{(1)}. \]
If the attacker manipulates \emph{only} pool~1, the aggregator 
sensitivity is \(\partial \hat p^{(1)}/\partial p_{1}^{(1)} = w_1\). 
To move \(\hat p^{(1)}\) by \(\Delta\), one needs \(\Delta p_{1}^{(1)} = \Delta/w_1\). 
Plugging \(p_{1}^{(1)}\mapsto p_{1}^{(1)}+\Delta/w_1\) into
\eqref{eq:y_in_single} or \eqref{eq:x_in_single} gives the 
required trade amounts per direction. We now turn to the question of how an attacker can influence the aggregated price at lowest cost. We start with the optimal attacker strategy for any given weights $w_{1}$ and $w_{2}$.

\paragraph{Optimal Manipulation.} Assume both pools start at the same
 fair price \(p_0=p^{\star}\), so \(p_1^{(0)}=p_2^{(0)}=p_0\). 
 Fix a target distortion factor \(r\ge 1\). Write \(p_i^{(1)}=s_i^2 p_0\) with \(s_i>0\). Achieving \(\hat p^{(1)}=r\,p_0\) at minimum cost reduces to the program
 \[
   \min_{s_1,s_2>0}\ y_1^{(0)}f(s_1^2)+y_2^{(0)}f(s_2^2)
   \quad\text{s.t.}\quad
   w_1 s_1^2+w_2 s_2^2 = r,
 \]
 where \(f(\cdot)\) is the single-pool factor from \eqref{eq:upward_cost}. A downward distortion \(\hat p^{(1)}=p_0/r\) corresponds to replacing \(r\) by \(1/r\); since \(f(r)=f(1/r)\), the minimal cost is the same in both directions.
At an interior optimum, the Lagrange first-order conditions take the form
\[
  y_i^{(0)}\, f'(t_i) = \lambda\, w_i,\qquad t_i:=s_i^2,
\]
for some multiplier \(\lambda>0\). Thus the marginal cost \(f'(t_i)\) is proportional to the weight-to-liquidity ratio \(w_i/y_i^{(0)}\). In particular, whenever the relevant multipliers lie in the convex range \(t_i\in(0,3]\) (so \(f'\) is increasing, see Figure \ref{fig:marginal_cost}), pools with larger \(w_i/y_i^{(0)}\) are pushed to larger \(t_i\), i.e., manipulated more aggressively.

	\begin{figure}[t]
	  \centering
	  \begin{tikzpicture}
	    \begin{axis}[
	      width=0.66\textwidth,
	      height=3.7cm,
	      xmin=1, xmax=10,
	      ymin=0, ymax=0.24,
	      axis x line=bottom,
	      axis y line=left,
      xlabel={Price multiplier \(t=p^{(1)}/p_0\)},
      ylabel={\(f'(t)\)},
      xtick={1, 3, 5, 7, 9},
      ytick={0, 0.05, 0.10, 0.15, 0.20},
      scaled y ticks=false,
      yticklabel style={/pgf/number format/fixed,/pgf/number format/precision=2},
      grid=major
    ]
      \addplot[domain=1:10, samples=200, thick, blue] {0.5*x^(-0.5) - 0.5*x^(-1.5)};
      \addplot[dashed, thick, red] coordinates {(3,0) (3,0.22)};
      \node[red, anchor=south west] at (axis cs:3,0.005) {$t=3$};
      \node[blue, anchor=south] at (axis cs:2.2,0.19) {increasing};
      \node[blue, anchor=north] at (axis cs:6,0.06) {decreasing};
    \end{axis}
  \end{tikzpicture}
  \caption{Marginal manipulation cost \(f'(t)=(t-1)/(2t^{3/2})\) for a single CPMM pool. Marginal cost peaks at \(t=3\) and declines thereafter, which makes concentrated mean attacks attractive once some pools are pushed beyond \(3\times\) their initial price.}
  \label{fig:marginal_cost}
\end{figure}
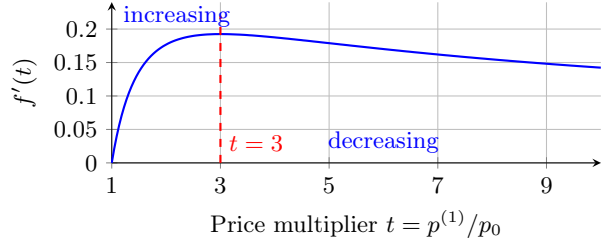

\subsubsection{Weighted Median }
Let \(p_{(1)}\le p_{(2)}\) be the sorted quotes and corresponding weights \(w_{(1)},w_{(2)}\) 
normalized so that \(w_{(1)}+w_{(2)}=1\). We use the \emph{lower} weighted median, defined as the smallest price level whose cumulative weight is at least $1/2$. Thus \(\tilde p=p_{(1)}\) if \(w_{(1)}\ge 1/2\) and \(\tilde p=p_{(2)}\) otherwise; under the exact tie \(w_{(1)}=w_{(2)}=1/2\), this convention returns \(\tilde p=p_{(1)}\).

 \paragraph{Optimal Manipulation.} With two pools, the weighted median is piecewise constant: the aggregator output is the price of whichever pool carries the majority weight. To move \(\tilde p\), the attacker must either (i) manipulate the dominating pool, or (ii) push the manipulated pool past the other so that the majority-weight index flips. In either subcase, the required trade sizes and costs follow from the single-pool formulas with the appropriate target price level.

\subsection{Multiple CPMM Pools}\label{sec:npool-exact}
We now turn to the general setting of $N$ independent pools. Throughout this subsection, let
\[
  y_{\mathrm{tot}} := \sum_{i=1}^N y_i^{(0)}
\]
denote the total quote depth across pools.
\vspace{-0.3cm} 
\subsubsection{Weighted Mean}
We first define the minimal cost of manipulation \(F_N(\bm w;r)\) for a specific choice of aggregation weights \(\bm w\) and target factor \(r\ge 1\). This is the economic loss an attacker must incur to achieve \(\hat p^{(1)}=r\,p_0\) when they optimize their trades across pools to minimize cost. With \(s_i\) representing the square root of the per-pool price multiplier (i.e., \(p_i^{(1)} = s_i^2 p_0\)), define
\[
  F_N(\bm w;r) \;:=\; \min_{\{s_i>0\}}\; \sum_{i=1}^N y_i^{(0)}\big(s_i+s_i^{-1}-2\big)\quad\text{s.t.}\quad \sum_{i=1}^N w_i s_i^2 = r.
\]
The oracle designer's problem is to choose \(\bm w\) to maximize this adversarial cost. The following theorem gives a small-distortion (second-order) characterization and identifies liquidity weights as locally max--min optimal for weighted-mean oracles near \(r=1\).

\begin{theorem}\label{thm:weighted-mean-opt}
Fix \(r=1+\varepsilon\) with \(\varepsilon\to 0^+\). Then, uniformly over weight vectors \(\bm w=(w_i)_{i=1}^N\) with \(w_i\ge 0\) and \(\sum_i w_i=1\),
\[
  F_N(\bm w;1+\varepsilon)
  \;=\;
  \frac{\varepsilon^2}{4}\,\frac{1}{\sum_{i=1}^N w_i^2/y_i^{(0)}} + o(\varepsilon^2).
\]
Consequently,
\[
  \sup_{w_i\ge 0,\,\sum_i w_i=1} F_N(\bm w;1+\varepsilon)
  \;=\;
  \frac{\varepsilon^2}{4}\,y_{\mathrm{tot}} + o(\varepsilon^2),
\]
and the leading-order term is uniquely maximized by the liquidity weights \(w_i^{\star}=y_i^{(0)}/y_{\mathrm{tot}}\).
\end{theorem}
\begin{proof}
Deferred to Appendix~\ref{app:weighted-mean-proof}.
\end{proof}

\paragraph{Convex regime (exact benchmark under a per-pool cap).}
Fix an upward target \(r\in[1,3]\) and consider attacks constrained by \(t_i:=p_i^{(1)}/p_0\in[1,3]\).
On this range \(f\) is convex. Under liquidity weights \(w_i^{\star}=y_i^{(0)}/y_{\mathrm{tot}}\), Jensen’s inequality implies
\[
  \sum_{i=1}^N y_i^{(0)} f(t_i) \;\ge\; y_{\mathrm{tot}}\, f(r),
\]
with equality at \(t_i\equiv r\). Since \(t_i\equiv r\) is feasible for every weight vector, no designer can force manipulation cost above \(y_{\mathrm{tot}}f(r)\), and thus liquidity weights are max--min optimal within this capped convex regime.
In the unconstrained model, by contrast, the aggregate condition \(r<3\) does not preclude concentrated attacks with some \(t_i>3\) when some weights are small.

\paragraph{Convexity threshold and concentrated attacks.}
The per-pool cost factor \(f(t)=\sqrt t+1/\sqrt t-2\) has a crucial inflection point at \(t=3\), since
\[
  f''(t)=\frac{3-t}{4t^{5/2}},
\]
so \(f\) is convex on \((0,3]\) and concave on \([3,\infty)\).
Equivalently, the marginal cost \(f'(t)=(t-1)/(2t^{3/2})\) peaks at \(t=3\) and decreases thereafter; see Fig.~\ref{fig:marginal_cost}.

This implies a qualitative shift in attacker incentives: attacks are dispersed when the relevant multipliers remain in \((0,3]\), but can become concentrated once some pool is pushed beyond \(t=3\). In particular, if \(r>3\) then any feasible \((t_i)_i\) must satisfy \(\max_i t_i\ge r>3\), and even if \(r<3\) a concentrated mean attack may still push some pools past \(t=3\) when some weights \(w_i\) are small.
We leave a full max--min characterization of optimal weighted-mean weights outside the small-distortion regime to future work; Appendix~\ref{app:weighted-mean-proof} gives a simple counterexample showing that liquidity weights need not be max--min optimal at large distortion levels and discusses alternative weight choices.

\subsubsection{Weighted Median}
Let $N\ge 2$ pools start at a common price $p_0$. Write $y_i^{(0)}$ for the quote reserve of pool $i$ and let $w_i>0$ be the aggregation weights normalized so that $\sum_i w_i=1$. Fix a target distortion factor $r\ge 1$ and define the per-pool cost factor
\[
 f(r) \;:=\; \sqrt{r}+\frac{1}{\sqrt{r}}-2,\qquad C_i(r) \;=\; y_i^{(0)} f(r).
\]
The weighted median $\widetilde p$ at $t$ is the smallest price level such that the cumulative weight at or below that level is at least $1/2$. With all quotes initially at $p_0$, to enforce an upward distortion $\widetilde p\ge r\,p_0$ it is necessary and sufficient to move a subset of pools $S\subset\{1,\dots,N\}$ to $r\,p_0$ so that their cumulative weight covers half the mass:
\begin{equation}\label{eq:wmed_cover}
 \sum_{i\in S} w_i \;\ge\; \tfrac{1}{2} \quad \Longleftrightarrow \quad \widetilde p\,\ge\, r\,p_0\,.
\end{equation}
Any pool not in $S$ can remain at $p_0$ (moving it to an intermediate price in $(p_0,rp_0)$ does not change the median). Hence the attacker’s problem reduces to the one-constraint covering program
\begin{equation}\label{eq:wmed_set_cover}
  C_{\text{med}}(r;\bm w) \;=\; \min_{S\subset\{1,\dots,N\}}\; \sum_{i\in S} y_i^{(0)} f(r) \quad \text{s.t.}\quad \sum_{i\in S} w_i \ge \tfrac{1}{2}.
\end{equation}
There is no closed form in general because the feasible set depends on the discrete weight configuration $\{w_i\}$. Nevertheless, the structure is simple and yields an explicit \emph{strategy}:

\paragraph{Optimal strategy (set form).} Because $C_{\text{med}}(r;\bm w)$ is additive across selected pools, it suffices to choose a subset $S$ with total weight at least $1/2$ that minimizes $\sum_{i\in S} y_i^{(0)}$. This ``minimum-cost cover'' is easy to compute for the pool counts that arise in practice, and its structure is transparent: optimal attacks prioritize pools with small depth per unit of weight, i.e., small ratios $y_i^{(0)}/w_i$. A simple greedy candidate is to sort pools by $y_i^{(0)}/w_i$ and add them in this order until the $1/2$ threshold is reached; this is exact in the two-pool and equal-weight cases below. In particular, there is never a reason to overshoot \(rp_0\) on any selected pool.

\paragraph{Special cases and bounds.}
\begin{enumerate}
  \item Two pools: assume without loss of generality that $w_1\ge w_2$. If $w_1>\tfrac{1}{2}$, then the weighted median equals the quote of pool~1, so an optimal attack sets $p_1=r p_0$ and leaves $p_2=p_0$, yielding $C^{\text{med}}=y_1^{(0)}f(r)=y_{\mathrm{major}}^{(0)}f(r)$. In the tie case $w_1=w_2=\tfrac{1}{2}$, our lower-median convention returns the smaller quote, so to enforce $\tilde p\ge r p_0$ one must move \emph{both} pools to $r p_0$, giving $C^{\text{med}}=y_{\mathrm{tot}}f(r)$. If instead ties are resolved by interpolation (e.g., $\tilde p=(p_{(1)}+p_{(2)})/2$), then in the tie case it also suffices to leave one pool at $p_0$ and move the other to $(2r-1)p_0$; choosing the cheaper pool yields cost $y_{\min}^{(0)} f(2r-1)$ with $y_{\min}^{(0)}:=\min(y_1^{(0)},y_2^{(0)})$. Hence, under midpoint interpolation, the tie-case cost is $\min\{y_{\mathrm{tot}} f(r),\,y_{\min}^{(0)} f(2r-1)\}$.
  \item Equal weights ($w_i=1/N$): let $k=\lceil N/2\rceil$ and order depths as $y_{(1)}^{(0)}\le\cdots\le y_{(N)}^{(0)}$. An optimal attack moves the $k$ pools with smallest depths to $r p_0$ (leaving the others at $p_0$), giving the closed form
  \[
    C^{\text{med}}(r)= f(r)\sum_{j=1}^k y_{(j)}^{(0)}.
  \]
  \item Liquidity weights ($w_i\propto y_i^{(0)}$): covering $1/2$ of the total weight requires at least half the total quote depth, so
\[ \tfrac{1}{2}\,y_{\mathrm{tot}}\;\le\;\min_{S:\,\sum w_i\ge 1/2}\sum_{i\in S} y_i^{(0)}\;\le\; \tfrac{1}{2}\,y_{\mathrm{tot}} + y_{\max}^{(0)}, \]
implying the median manipulation cost satisfies the exact bounds
\[ \tfrac{1}{2}\,f(r)\,y_{\mathrm{tot}}\;\le\; C_{\text{med}}(r;\bm w^{\star})\;\le\; f(r)\Big(\tfrac{1}{2}y_{\mathrm{tot}} + y_{\max}^{(0)}\Big). \]
For comparison, for the weighted \emph{mean} with the same liquidity weights, Theorem~\ref{thm:weighted-mean-opt} yields the small-distortion expansion $F_N(\bm w^{\star};1+\varepsilon) = (\varepsilon^2/4)\,y_{\mathrm{tot}} + o(\varepsilon^2)$. Table~\ref{tab:summary} summarizes the contrast between mean- and median-based aggregation across distortion regimes.
\end{enumerate}

We now show that, within the class of weighted-median oracles, liquidity weighting maximizes the attacker’s minimal cost.
\begin{theorem}\label{thm:median-opt-weights}
Let $\bm y=(y_1^{(0)},\dots,y_N^{(0)})$. Define
\[ \Theta(\bm y) := \min\Bigl\{\sum_{i\in S} y_i^{(0)} : S\subset\{1,\dots,N\},\; \sum_{i\in S} y_i^{(0)}\ge \tfrac{1}{2}y_{\mathrm{tot}}\Bigr\}. \]
Then for any $r\ge 1$,
\[ \sup_{\bm w\ge 0,\,\sum w_i=1} C_{\mathrm{med}}(r;\bm w) = f(r)\,\Theta(\bm y), \]
attained by the liquidity weights $w_i^{\star}=y_i^{(0)}/y_{\mathrm{tot}}$. Moreover, $\tfrac{1}{2}y_{\mathrm{tot}}\le \Theta(\bm y)\le \tfrac{1}{2}y_{\mathrm{tot}}+y_{\max}^{(0)}$, and for $N=2$, $\Theta(\bm y)=\max(y_1^{(0)},y_2^{(0)})$.
\end{theorem}
\begin{proof}
Deferred to Appendix~\ref{app:weighted-median-proof}.
\end{proof}

\begin{table}[t]
\centering
\caption{Summary of manipulation costs and optimal weight designs. Here $t_i:=p_i^{(1)}/p_0$ is the per-pool price multiplier, and $\Theta(\bm y)$ is the minimum depth of a majority-weight subset.}
\label{tab:summary}
\small
\setlength{\tabcolsep}{4pt}
\renewcommand{\arraystretch}{1.3}
\begin{tabular}{@{}>{\raggedright\arraybackslash}p{2.6cm}>{\raggedright\arraybackslash}p{5.6cm}>{\raggedright\arraybackslash}p{3.0cm}@{}}
\toprule
\textbf{Scenario} & \textbf{Designer Takeaway (Max--Min)} & \textbf{Minimal Cost} \\ \midrule
\multicolumn{3}{l}{\textit{Independent Pools (No Arbitrage)}} \vspace{2pt} \\
\textbf{Weighted Mean} &
\textbf{Local Optimality ($r \to 1$):} Liquidity weights $w_i \propto y_i^{(0)}$ are optimal. \newline
\textbf{Large Distortions:} Fragile. Once some $t_i>3$, marginal costs decrease and concentrated attacks can dominate; no uniform max--min weights. &
Small distortions: $\propto y_{\mathrm{tot}}$;\newline
large distortions: sublinear \\ \addlinespace[4pt]
\textbf{Weighted Median} &
\textbf{Global Optimality ($\forall r \ge 1$):} Liquidity weights $w_i \propto y_i^{(0)}$ are optimal for \emph{any} target distortion. \newline
Robust: requires manipulating a majority-weight subset. &
$f(r)\,\Theta(\bm y)$ \\ \midrule
\multicolumn{3}{l}{\textit{Perfect Cross-Pool Arbitrage}} \vspace{2pt} \\
\textbf{Any Symmetric Aggregator} &
\textbf{Irrelevance:} Arbitrage equalizes terminal prices $p_i \to p_{\mathrm{tar}}$. Aggregation rule does not affect cost. &
$f(r)\,y_{\mathrm{tot}}$ \\
\bottomrule
\end{tabular}
\end{table}

\subsection{Arbitrage Across Pools}
We now consider frictionless, immediate arbitrage between CPMM pools that all trade the same pair. Any configuration in which some pools are more distorted than others creates an immediate arbitrage opportunity: an arbitrageur can buy in the cheap pool and sell in the expensive one until prices equalize, earning additional profit at the attacker’s expense and partially undoing the oracle distortion. Hence, in any minimal-cost attack under perfect arbitrage, we may restrict attention to equalized terminal configurations in which all pools end at the same target price \(p_{\mathrm{tar}}=r\,p_0\) (or \(p_{\mathrm{tar}}=p_0/r\)).

Assume the system starts at a common efficient price \(p_0\), so \(y_i^{(0)}/x_i^{(0)}=p_0\) and \(k_i=x_i^{(0)}y_i^{(0)}\). In the equalized state, each pool is shifted by the same price multiplier, so the total economic loss is the sum of the single-pool costs, yielding
\begin{equation}\label{eq:arb_min_cost}
  C^{\star}(r) = y_{\mathrm{tot}}\, f(r).
\end{equation}
Thus, under perfect cross-pool arbitrage, the \(N\) CPMMs behave like a single effective pool whose quote reserve equals the total depth \(y_{\mathrm{tot}}\): the cost \eqref{eq:arb_min_cost} is exactly the single-pool expression with \(y_0\) replaced by \(y_{\mathrm{tot}}\). Because equalization forces all post-arbitrage quotes to coincide, any symmetric aggregator (weighted mean, median, trimmed mean, etc.) returns \(p_{\mathrm{tar}}\), so the cost \eqref{eq:arb_min_cost} is independent of the particular symmetric aggregation rule. 

\section{Multi-Asset Extension}\label{sec:multi-asset}
We briefly sketch a multi-asset extension. A particularly clean setting is a \emph{star architecture} around a numéraire, where each non-numéraire asset is priced by aggregating CPMM pools that trade against the numéraire and manipulation cost is measured in numéraire units. Appendix~\ref{app:multi-asset-star} states and proves that, in such architectures, liquidity-proportional weights remain optimal asset-by-asset (exactly for weighted medians and locally for weighted means as \(r_a\to 1\)). Extending this analysis to general CPMM graphs with cross pairs is left for future work.
\section{Incorporating Swap Fees}\label{sec:fees}

Incorporating proportional swap fees into our CPMM cost formulas is straightforward.
With an input fee \(\phi\), only a fraction \((1-\phi)\) of the gross input is credited to the pool, so achieving a fixed distortion factor \(r\) requires gross trade sizes larger by a factor \(1/(1-\phi)\) than in \eqref{eq:y_in_single}–\eqref{eq:x_in_single}. In our one-shot formulas, this acts as a simple multiplicative adjustment of the direct manipulation cost.
In multi-pool settings, fees also widen classical no-arbitrage bands for cross-pool cycles (see \cite{angeris2020improved}), so small cross-pool discrepancies can persist because arbitrage is unprofitable unless they exceed the fee wedge. This weakens corrective arbitrage and can reduce the gross volume required to sustain a discrepancy, even though each swap pays fees.
Finally, implementation costs such as gas fees contribute an additive term per transaction.
Since our goal is a clean depth-driven baseline, we focus on the zero-fee case and treat the interaction between fees, arbitrage efficiency, and gas costs as deployment-specific refinements.

\section{Related Work}\label{sec:related-work}

\paragraph{CFMM microstructure and AMM oracles.}
The Uniswap v2 and v3 whitepapers~\cite{uniswapv2,uniswapv3} and the CFMM framework of Angeris and Chitra~\cite{angeris2020improved} describe invariant-based pricing, depth, and arbitrage for production CPMMs.
Subsequent work analyzes LP risk, fee design, and predictable losses~\cite{bergault2022amm,algebra2022fee,cartea2025predictable}, multi-token AMMs and closed-form arbitrage in $N$-asset pools~\cite{yang2024ntoken}, and routing and coupling effects across CFMM pools~\cite{angeris2022routing,sterrett2025cfmmcoupling}, as well as providing broader surveys of DeFi AMMs and CFMM mechanics~\cite{cartea2025defiamms}.
These papers treat slippage and arbitrage primarily as descriptive properties or sources of LP risk; to the best of our knowledge, none defines a general, closed‑form cost of manipulation metric of the form ``minimal loss to move the price by a given factor'' nor poses a max–min defender problem over cross‑pool weights.

\paragraph{Oracle manipulation, TWAPs, and DeFi security.}
Empirical and systems work documents oracle deviations and attacks in DeFi~\cite{zhang2020firstlook,qin2021attacks,yang2021flashloan}, and surveys oracle architectures and TWAP designs~\cite{aspembitova2022oracles,zhao2022trustworthy,deng2024defioracles}.
Uniswap v3 TWAP studies~\cite{adams2022uniswap,chaos2023twap} compute capital requirements for TWAP manipulation, while large-scale evaluations of Chainlink and cross-chain oracles~\cite{nadler2025chainlink,gansauer2025accuracy} and analysis frameworks such as OVer~\cite{luu2024over,boe2024over} provide risk metrics and stress tests under adversarial inputs.
These contributions quantify attack costs for specific oracle mechanisms (primarily arithmetic TWAPs) and propose mitigations such as time windows and circuit breakers, but most treat AMM pricing as a black box or work numerically with particular TWAP implementations; they do not derive closed-form  manipulation costs as explicit functions of liquidity and distortion, nor do they analyze a general multi-pool attacker--designer game over aggregation weights.

\paragraph{Robust aggregation, liquidity weights, and positioning.}
Robust statistics offers general tools for aggregation under outliers~\cite{hampel2001robust,huber1981robust,maronna2006robust}, and recent crypto-specific work~\cite{allouche2024crypto} derives nonasymptotic error bounds for weighted means and medians applied to exchange price data.
Industry indices such as the SIX Crypto Indices~\cite{six2020cryptoindices} implement volume- or liquidity-weighted medians to down-weight small venues.
These works justify liquidity and volume weights from a statistical-error perspective under contamination models, but do not model CFMM microstructure or a strategic attacker who must trade against bonding curves, and, to the best of our knowledge, they contain no CPMM-aware max--min optimality result for liquidity weights.
Our contribution is deliberately basic: we take the standard CPMM price curve, define a general, closed-form cost of manipulation metric for single and multiple CPMM pools, and then solve the associated attacker–designer problems for weighted means, weighted medians, and a multi-asset star architecture.
To our knowledge, this is the first work to show that, in a CPMM microstructure-aware setting, liquidity weights are max--min optimal for weighted medians (for any distortion level) and locally max--min optimal for weighted means near \(r=1\), and to show by counterexample that weighted means admit no distortion-uniform max--min weights beyond small distortions.

\section{Discussion and Future Work}\label{sec:discussion}
\paragraph{Multi-pool aggregation and on-chain feasibility}
In the independent-pool regime, our analysis shows that liquidity-weighted medians maximize the minimal cost of manipulation for any distortion level, while liquidity-weighted means are locally optimal near \(r=1\) but can be substantially more fragile under large distortions; under perfect cross-pool arbitrage the effective depth entering the cost formulas is simply the total quote reserve across pools.
Robustness is therefore primarily driven by how much CPMM depth backs the oracle and how that depth is distributed across pools, rather than by finer choices among symmetric aggregators.

\paragraph{Dwell-time robustness and protocol interaction}
The static cost of manipulation \(\mathrm{Cost}_A(r)\) studied in the main text can be embedded into a dynamic setting via the dwell-time extension \(\mathrm{Cost}_A(r,\tau)\) (Appendix~\ref{dwell}).
For any given application, one can compute or upper bound the maximal exploitable gain from a mispricing by a factor \(r\) maintained for a dwell \(\tau\); denote this by \(B(r,\tau)\) (e.g., the largest profit from shifting a liquidation threshold or triggering a mispriced payoff).
The deterrence criterion
\[
  \mathrm{Cost}_A(r,\tau) \;\gg\; B(r,\tau)
\]
then provides a quantitative notion of ``economic safety margin''.
Our CPMM cost formulas identify the per-block building blocks entering \(\mathrm{Cost}_A(r,\tau)\); deriving sharp multi-block lower bounds in concrete latency and congestion-control models, and matching them against protocol-specific \(B(r,\tau)\), is an important avenue for future work.

\paragraph{Aggregation rules under different fault models}

Our aggregation results highlight an important difference between CPMM-based on-chain oracles and the exchange-based setting studied by Allouche et al.~\cite{allouche2024crypto}, where trimmed medians are optimal under contamination.
For our economic cost of manipulation metric on \emph{on-chain} CPMM data, liquidity-weighted means are more expensive to manipulate than liquidity-weighted medians for small distortions in the independent-pool model, reflecting that means ``use'' all pools while medians only require a majority-weight cover. However, because per-pool CPMM manipulation cost grows sublinearly once a pool is pushed beyond the inflection threshold of the single-pool factor, weighted means can become significantly more fragile under large distortions; in this regime, designer-optimal mean weights can depend on the target distortion and there is no distortion-uniform max--min choice in general. Weighted medians retain a large-distortion guarantee: moving the oracle requires manipulating a subset covering at least half of the aggregation weight, regardless of how large \(r\) is.
Smart-contract bugs, misconfigured pools, or governance attacks can still create persistently faulty on-chain venues.
In such cases, median- or trimmed-mean aggregation over pools may be more appropriate to discount structurally broken pools while still using liquidity weights within the remaining set. Liquidity-weighted means remain natural when all CPMMs are correct and manipulation occurs only via trading.

\paragraph{Beyond CPMMs}

In practice, concentrated-liquidity AMMs (CLMMs) such as Uniswap v3 are widely used. A convenient reduced-form view is that they induce a state-dependent effective depth \(y_{\mathrm{eff}}(p)\) aggregating all active liquidity at price \(p\).
The cost of moving the price from \(p_0\) to \(p_1=r\,p_0\) (or \(p_1=p_0/r\)) is then obtained by integrating the CPMM slippage formulas along the path in price with \(y_0\) replaced locally by \(y_{\mathrm{eff}}(p)\), preserving the convexity and monotonicity properties that underpin our multi-pool optimization at the expense of simple closed forms in \(r\).
Extending our results to CLMMs requires modeling the effective depth profile \(y_{\mathrm{eff}}(p)\); we leave a full treatment of CLMM microstructure and tick dynamics to future work.

\section{Conclusion}

We introduced a quantitative notion of cost of manipulation for AMM-based oracles and analyzed how it depends on liquidity depth, aggregation rules, and arbitrage connectivity.
For a single CPMM pool, we recalled closed-form formulas for the trade size and economic loss required to move prices by a factor \(r\) (upward or downward).
In independent multi-pool settings, we solved the attacker–designer game for weighted medians (all distortion levels) and for weighted means locally near \(r=1\). We also showed by counterexample that weighted means admit no distortion-uniform optimal weights beyond small distortions, while under frictionless cross-pool arbitrage the cost collapses to simple total-depth expressions that are independent of the particular symmetric aggregator.
For weighted means, we also highlighted the inflection of the single-pool cost factor at multiplier \(t=3\), which induces a qualitative shift from dispersed to concentrated optimal attacks and motivates either median-type aggregation or distortion-aware weighting that down-weights shallow pools more aggressively in the large-distortion regime.
We extended this analysis to a multi-asset star architecture and proved an analogous optimality result for per-asset weights.
Finally, the dwell-time and rate-limit extensions illustrate how these static cost benchmarks can be combined with chain- and application-specific models to design AMM-based oracles whose manipulation cost dominates the economic gains available from induced mispricings.
\newpage

\bibliographystyle{abbrv}
\bibliography{references}

\appendix

\section{Sustained Manipulation and Protocol Constraints}\label{dwell}

The one-step cost of manipulation considered above abstracts away from timing and protocol-level limits. In practice, block structure, rate limits on shared objects, and latency all constrain both attackers and arbitrageurs and motivate a dynamic extension of our metric and of the comparison with application-specific benefit functions \(B(r,\tau)\).

\paragraph{Sustained distortions with dwell time.}
Let time be indexed by blocks \(t=0,1,\dots\), and let \(A_t\) be the aggregation rule at time \(t\), possibly incorporating time-windowed statistics (e.g., TWAPs). The attacker submits a sequence of trades \(u_0,\dots,u_{T-1}\), incurring cumulative loss \(C(u_{0:T-1})\). For a dwell parameter \(\tau\), a natural extension of our robustness metric is the minimal cost needed to keep the oracle distorted by a factor of at least \(r\) for \(\tau\) consecutive blocks:
\[
  D_t := \max\Big\{\frac{\hat p_t}{p_t^{\star}},\,\frac{p_t^{\star}}{\hat p_t}\Big\}.
\]
\[
  \mathrm{Cost}_A(r,\tau;T)
  := \inf\Big\{ C(u_{0:T-1}) : \exists t_0\ \text{s.t. }D_t\ge r\ \forall t\in[t_0,t_0+\tau-1]\Big\},
\]
with \(\hat p_t = A_t(\text{quotes after }u_0,\dots,u_t)\) and \(p_t^{\star}\) the efficient price process. The asymptotic cost \(\mathrm{Cost}_A(r,\tau)\) is defined via a liminf as \(T\to\infty\). Our static results identify the per-block building blocks entering \(\mathrm{Cost}_A(r,\tau)\); deriving sharp multi-block lower bounds in specific latency models, and matching them against protocol-specific upper bounds \(B(r,\tau)\), is a natural direction for further work and for quantitative risk budgeting.

\paragraph{Rate limits and shared-object congestion (Sui example).}
On high-throughput platforms with parallel execution such as Sui, contention on shared objects (e.g., a CPMM pool) becomes a bottleneck, so protocols rate-control how many transactions per checkpoint may touch the same shared object.
This interacts with the dwell-time notion of \(\mathrm{Cost}_A(r,\tau)\) by limiting how quickly arbitrageurs can respond to distortions and how many manipulative trades an attacker can sustain per block.
A detailed analysis of such rate limits—combining our static cost formulas with models of block-level access constraints, spam behavior, and explicit \(B(r,\tau)\) for concrete lending and derivatives protocols—is left for future work and is particularly relevant for Sui-style shared-object architectures.

\section{Derivations for Single-Pool Formulas}\label{app:single-pool-derivations}

We briefly derive the baseline single-pool manipulation formulas used throughout the paper.

\subsection{Trade needed to hit a relative price target}

Consider a CPMM with reserves \((x_0,y_0)\), invariant \(k=x_0y_0\) and initial price \(p_0=y_0/x_0\). Fix a target factor \(r\ge 1\). For an upward move we target \(p_1=r\,p_0\), and for a downward move we target \(p_1=p_0/r\). Post-trade reserves \((x_1,y_1)\) must satisfy
\[
  p_1 = \frac{y_1}{x_1},\qquad x_1 y_1 = k.
\]
Solving these two equations yields
\[
  x_1 = \sqrt{\frac{k}{p_1}},\qquad y_1 = \sqrt{k\,p_1}.
\]
Using \(y_0=\sqrt{k p_0}\) and \(x_0=\sqrt{k/p_0}\), write \(s=\sqrt r\). Then:
\begin{itemize}[leftmargin=*]
  \item Upward move (\(p_1=r\,p_0\)): \(y_1=y_0 s\), so \(y^{\mathrm{in}}_{\uparrow}(r)=y_1-y_0=y_0(s-1)\), which is \eqref{eq:y_in_single}.
  \item Downward move (\(p_1=p_0/r\)): \(x_1=x_0 s\), so \(x^{\mathrm{in}}_{\downarrow}(r)=x_1-x_0=x_0(s-1)\), which is \eqref{eq:x_in_single}.
\end{itemize}
These expressions make explicit that trade size scales like \(\sqrt{k}\) and depends on the target only through \(r\).

\subsection{Economic cost at \(p^{\star}=p_0\)}

We mark attacker losses to market at a reference price \(p^{\star}\) representing the external efficient value. For the baseline case \(p^{\star}=p_0\):
\begin{itemize}[leftmargin=*]
  \item \emph{Upward move.} A \(Y\!\to\!X\) trade of size \(y^{\mathrm{in}}_{\uparrow}\) sends reserves from \((x_0,y_0)\) to \((x_1,y_1)\) with \(x_1y_1=k\). The attacker receives
  \[
    x^{\mathrm{out}} = x_0 - x_1 = x_0 - \frac{k}{y_0 + y^{\mathrm{in}}_{\uparrow}}.
  \]
  Valued at \(p^{\star}=p_0\), the loss in \(Y\) units is
  \[
    C_Y^{\uparrow} = y^{\mathrm{in}}_{\uparrow} - p_0 x^{\mathrm{out}}.
  \]
  Substituting \(y^{\mathrm{in}}_{\uparrow}=y_0(s-1)\), \(x^{\mathrm{out}}=x_0(1-1/s)\) and \(p_0 x_0=y_0\) yields
  \[
    C_Y^{\uparrow}(r) = y_0\!\left(s+\tfrac{1}{s}-2\right),
  \]
  which is \eqref{eq:upward_cost}.
  \item \emph{Downward move.} An \(X\!\to\!Y\) trade of size \(x^{\mathrm{in}}_{\downarrow}\) yields
  \[
    y^{\mathrm{out}} = y_0 - y_1 = y_0 - \frac{k}{x_0 + x^{\mathrm{in}}_{\downarrow}}.
  \]
  The loss in \(Y\) units is
  \[
    C_Y^{\downarrow} = p_0 x^{\mathrm{in}}_{\downarrow} - y^{\mathrm{out}}.
  \]
  With \(x^{\mathrm{in}}_{\downarrow}=x_0(s-1)\), \(y^{\mathrm{out}}=y_0(1-1/s)\) and \(p_0 x_0 = y_0\), we obtain
  \[
    C_Y^{\downarrow}(r) = y_0\!\left(s+\tfrac{1}{s}-2\right),
  \]
  which is \eqref{eq:downward_cost}.
\end{itemize}
In both directions, \(C_Y\) coincides with the standard notion of slippage loss (difference between what the attacker pays and the fair value of what they receive) used in CFMM analyses such as \cite{angeris2020improved}.

\section{Weighted-Mean Multi-Pool Optimization}\label{app:weighted-mean-proof}

In this appendix we provide a detailed proof of Theorem~\ref{thm:weighted-mean-opt}.
We work with price multipliers $t_i:=s_i^2$ (so that $t_i= p_i^{(1)}/p_0$) and write the attacker’s problem as
\[
  F_N(\bm w;r)
  = \min\Big\{\sum_{i=1}^N y_i^{(0)} f(t_i) : t_i>0,\ \sum_{i=1}^N w_i t_i = r\Big\},
  \qquad
  f(t)=\sqrt{t}+\frac{1}{\sqrt{t}}-2.
\]

\paragraph{Why we focus on small distortions.}
The key technical feature is that the per-pool factor $f$ is not globally convex: a direct calculation gives
\[
  f''(t)=\frac{3-t}{4t^{5/2}},
\]
so $f$ is convex on $(0,3]$ and concave on $[3,\infty)$.
Equivalently, the marginal cost $f'(t)=(t-1)/(2t^{3/2})$ increases on $(0,3]$ and decreases on $[3,\infty)$; see Fig.~\ref{fig:marginal_cost}.
When some pools are pushed beyond the inflection point $t=3$, the attacker can benefit from concentrating distortion on a small-weight pool while leaving most pools close to $t=1$; this may occur even for moderate aggregate targets $r$ if some $w_i$ are small.
For a concrete illustration, take $N=2$ pools with quote depths $(y_1^{(0)},y_2^{(0)})=(1,M)$ and liquidity weights $w^{\star}=(\tfrac{1}{M+1},\tfrac{M}{M+1})$. Fix any target $r>1$ and set $t_2=1$ and $t_1=1+(r-1)(M+1)$; then
\[
  w_1^{\star}t_1+w_2^{\star}t_2=\frac{1}{M+1}\big(1+(r-1)(M+1)\big)+\frac{M}{M+1}\cdot 1=r,
\]
so the attack is feasible at cost $F_2(w^{\star};r)\le f(1+(r-1)(M+1))$, while the pooled benchmark costs $(M+1)f(r)$.
	Since $f(t)\le \sqrt{t}$ for all $t\ge 1$ and $f(t)\sim \sqrt{t}$ as $t\to\infty$, the ratio between these two costs is $O(M^{-1/2})$ as $M\to\infty$.
	Thus the pooled-liquidity benchmark $y_{\mathrm{tot}}f(r)$ cannot hold uniformly in $r$ away from $r=1$.
	Moreover, the same example shows that liquidity weights need not be max--min optimal for weighted means at large distortion levels: if the designer instead ignores the shallow pool and sets weights $\widetilde w=(0,1)$, then the oracle depends only on pool~2 and any successful attack must set $t_2=r$, incurring cost $F_2(\widetilde w;r)=M f(r)$. For fixed $r>1$ (in particular for $r>3$), this scales like $\Theta(M)$, while the feasible concentrated attack above achieves cost at most $f(1+(r-1)(M+1))=O(\sqrt{M})$ under liquidity weights. Hence, for $M$ sufficiently large, $\widetilde w$ yields strictly larger minimal manipulation cost than $w^{\star}$.
	Accordingly, Theorem~\ref{thm:weighted-mean-opt} is stated as a small-distortion result.

\paragraph{A concentration bound and a large-distortion design heuristic.}
For any weights $\bm w$ and any target $r\ge 1$, a feasible attack is to leave all pools at $t_j=1$ except one pool~$i$, and set $t_i=1+(r-1)/w_i$; this yields
\[
  F_N(\bm w;r) \;\le\; \min_i y_i^{(0)} f\!\Big(1+\frac{r-1}{w_i}\Big).
\]
Since $f(t)\sim \sqrt{t}$ as $t\to\infty$, this upper bound behaves like $\sqrt{r-1}\,\min_i y_i^{(0)}/\sqrt{w_i}$ for large distortions, suggesting quadratic liquidity weights $w_i\propto (y_i^{(0)})^2$ as a conservative way to equalize the cost of such single-pool attacks.

\subsection*{Second-order expansion and attacker optimum}

Fix weights $w_i\ge 0$ with $\sum_i w_i=1$ and set $r=1+\varepsilon$ with $\varepsilon\to 0^+$.
Write $t_i=1+\delta_i$ with $\delta_i>-1$, so the constraint becomes
\[
  \sum_{i=1}^N w_i \delta_i = \varepsilon.
\]
Since the constant choice $\delta_i\equiv\varepsilon$ is feasible, we have $F_N(\bm w;1+\varepsilon)\le y_{\mathrm{tot}} f(1+\varepsilon)=O(\varepsilon^2)$ uniformly in $\bm w$.
Because $f(1+\delta)\to\infty$ as $\delta\downarrow -1$ or $\delta\to\infty$, any minimizer must satisfy $\max_i|\delta_i|\to 0$ as $\varepsilon\to 0$ (otherwise the objective would be bounded below by a positive constant).
Hence we may use the Taylor expansion
\[
  f(1+\delta) = \frac{\delta^2}{4} + O(\delta^3)\qquad(\delta\to 0),
\]
to obtain
\[
  F_N(\bm w;1+\varepsilon)
  = \frac{1}{4}\,\min\Big\{\sum_{i=1}^N y_i^{(0)}\delta_i^2 : \sum_{i=1}^N w_i\delta_i=\varepsilon\Big\} + o(\varepsilon^2).
\]
The quadratic program is solved by Cauchy--Schwarz:
\[
  \varepsilon^2
  = \Big(\sum_{i=1}^N w_i\delta_i\Big)^2
  \le \Big(\sum_{i=1}^N \frac{w_i^2}{y_i^{(0)}}\Big)\Big(\sum_{i=1}^N y_i^{(0)}\delta_i^2\Big),
\]
with equality iff $\delta_i\propto w_i/y_i^{(0)}$.
Therefore,
\[
  \min\Big\{\sum_{i=1}^N y_i^{(0)}\delta_i^2 : \sum_{i=1}^N w_i\delta_i=\varepsilon\Big\}
  = \frac{\varepsilon^2}{\sum_{i=1}^N w_i^2/y_i^{(0)}},
\]
and thus
\[
  F_N(\bm w;1+\varepsilon)
  = \frac{\varepsilon^2}{4}\,\frac{1}{\sum_{i=1}^N w_i^2/y_i^{(0)}} + o(\varepsilon^2).
\]

\subsection*{Designer optimum}

To maximize the leading-order term, the oracle designer minimizes $\sum_i w_i^2/y_i^{(0)}$ over all weights $w_i\ge 0$ with $\sum_i w_i=1$.
By Cauchy--Schwarz,
\[
  1 = \Big(\sum_{i=1}^N w_i\Big)^2
  \le \Big(\sum_{i=1}^N \frac{w_i^2}{y_i^{(0)}}\Big)\Big(\sum_{i=1}^N y_i^{(0)}\Big)
  = y_{\mathrm{tot}}\sum_{i=1}^N \frac{w_i^2}{y_i^{(0)}},
\]
with equality iff $w_i\propto y_i^{(0)}$.
Hence $\sum_i w_i^2/y_i^{(0)}\ge 1/y_{\mathrm{tot}}$, and the leading-order cost is uniquely maximized by the liquidity weights $w_i^{\star}=y_i^{(0)}/y_{\mathrm{tot}}$, giving
\[
  \sup_{w_i\ge 0,\,\sum_i w_i=1} F_N(\bm w;1+\varepsilon)
  = \frac{\varepsilon^2}{4}\,y_{\mathrm{tot}} + o(\varepsilon^2).
\]

\section{Weighted-Median Multi-Pool Optimization}\label{app:weighted-median-proof}

We now prove the weighted-median weight-design result stated in Theorem~\ref{thm:median-opt-weights}.
 Recall the setup of the $N$-pool median: for quote reserves 
 $y_i^{(0)}>0$ and weights $w_i\ge 0$ with $\sum_i w_i=1$, 
 the attacker’s minimal cost to enforce a distortion factor $r\ge 1$ is
\[
  C^{\mathrm{med}}(r;w) = f(r)\, m(w),
\]
\[
  m(w) := \min\Big\{\sum_{i\in S} y_i^{(0)} : S\subset\{1,\dots,N\},\ \sum_{i\in S} w_i \ge \tfrac{1}{2}\Big\}.
\]
Thus the oracle designer solves
\[
  \sup_{w\in\Delta_N} C^{\mathrm{med}}(r;w)
  = f(r)\,\sup_{w\in\Delta_N} m(w),
  \qquad \Delta_N = \{w_i\ge 0,\ \sum_i w_i=1\}.
\]
We therefore focus on the purely combinatorial quantity $m(w)$.

Let $\bm y=(y_1^{(0)},\dots,y_N^{(0)})$ and $y_{\mathrm{tot}}=\sum_i y_i^{(0)}$. Define
\[
  \Theta(\bm y)
  := \min\Big\{\sum_{i\in S} y_i^{(0)} : S\subset\{1,\dots,N\},\ \sum_{i\in S} y_i^{(0)} \ge \tfrac{1}{2}y_{\mathrm{tot}}\Big\},
\]
i.e., the smallest total depth carried by any subset whose total depth is at least half of $y_{\mathrm{tot}}$.

\begin{lemma}\label{lem:theta-bounds}
For any $\bm y$ as above,
\[
  \tfrac{1}{2}y_{\mathrm{tot}} \;\le\; \Theta(\bm y) \;\le\; \tfrac{1}{2}y_{\mathrm{tot}} + y_{\max}^{(0)},
\]
where $y_{\max}^{(0)} := \max_i y_i^{(0)}$. Moreover, when $N=2$ we have $\Theta(\bm y)=\max(y_1^{(0)},y_2^{(0)})$.
\end{lemma}
\begin{proof}
By definition, every admissible subset $S$ satisfies $\sum_{i\in S} y_i^{(0)} \ge \tfrac{1}{2}y_{\mathrm{tot}}$, so the minimum is at least $\tfrac{1}{2}y_{\mathrm{tot}}$. For the upper bound, sort indices so that $y_{(1)}^{(0)}\le\cdots\le y_{(N)}^{(0)}$ and let $k$ be the smallest index with $\sum_{i=1}^k y_{(i)}^{(0)}\ge\tfrac{1}{2}y_{\mathrm{tot}}$. Then $\Theta(\bm y) = \sum_{i=1}^k y_{(i)}^{(0)}$ and
\[
  \Theta(\bm y)
  \le \tfrac{1}{2}y_{\mathrm{tot}} + y_{(k)}^{(0)}
  \le \tfrac{1}{2}y_{\mathrm{tot}} + y_{\max}^{(0)}.
\]
When $N=2$, either $y_1^{(0)}\ge\tfrac{1}{2}y_{\mathrm{tot}}$ or $y_2^{(0)}\ge\tfrac{1}{2}y_{\mathrm{tot}}$ (or both), and the minimal subset achieving the half-depth threshold is the index with larger depth, so $\Theta(\bm y)=\max(y_1^{(0)},y_2^{(0)})$.
\end{proof}

\begin{lemma}\label{lem:median-cost-characterization}
For any $w\in\Delta_N$ and $r>0$,
\[
  C_{\mathrm{med}}(r;w) = f(r)\, m(w),
\]
with $m(w)$ as defined above.
\end{lemma}
\begin{proof}
By definition of the weighted median, moving the median from $p_0$ to $r p_0$ requires the attacker to select a subset $S$ of pools whose cumulative weight is at least $1/2$ and move exactly those pools to $r p_0$; moving any pool to a price in $(p_0,rp_0)$ does not change the median level. Since all pools start at $p_0$ and each moved pool must be at $r p_0$, the attack cost is the sum of per-pool manipulation costs $C_i(r)=y_i^{(0)}f(r)$ over $i\in S$. Minimizing over all subsets with $\sum_{i\in S} w_i\ge 1/2$ yields precisely the expression for $m(w)$.
\end{proof}

\begin{proposition}\label{prop:median-weights-appendix}
For any $\bm y$ as above and any $r\ge 1$,
\[
  \sup_{w\in\Delta_N} C_{\mathrm{med}}(r;w)
  = f(r)\,\Theta(\bm y),
\]
attained by the liquidity weights $w_i^{\star}=y_i^{(0)}/y_{\mathrm{tot}}$. In particular,
\[
  \sup_{w\in\Delta_N} m(w) = \Theta(\bm y).
\]
\end{proposition}
\begin{proof}
By Lemma~\ref{lem:median-cost-characterization} it suffices to study $\sup_{w\in\Delta_N} m(w)$.

\emph{Lower bound and attainment.} Take liquidity weights $w_i^{\star} = y_i^{(0)}/y_{\mathrm{tot}}$. For any subset $S$,
\[
  \sum_{i\in S} w_i^{\star} \;\ge\; \tfrac{1}{2}
  \quad\Longleftrightarrow\quad
  \sum_{i\in S} y_i^{(0)} \;\ge\; \tfrac{1}{2}y_{\mathrm{tot}}.
\]
Thus
\[
  m(w^{\star})
  = \min\Big\{\sum_{i\in S} y_i^{(0)} : \sum_{i\in S} y_i^{(0)}\ge\tfrac{1}{2}y_{\mathrm{tot}}\Big\}
  = \Theta(\bm y),
\]
so $\sup_{w} m(w)\ge m(w^{\star})=\Theta(\bm y)$ and therefore
\[
  \sup_{w} C_{\mathrm{med}}(r;w)
  \;\ge\; f(r)\,\Theta(\bm y).
\]

\emph{Upper bound.} Fix an arbitrary $w\in\Delta_N$. By definition of $\Theta(\bm y)$ there are only finitely many candidate subsets $S\subset\{1,\dots,N\}$, so the minimum in its definition is attained. Let $S^{\theta}$ be any subset achieving this minimum, i.e.,
\[
  \sum_{i\in S^{\theta}} y_i^{(0)} = \Theta(\bm y),\qquad
  \sum_{i\in S^{\theta}} y_i^{(0)} \ge \tfrac{1}{2}y_{\mathrm{tot}}.
\]
There are two cases.

\begin{itemize}[leftmargin=*]
  \item If $\sum_{i\in S^{\theta}} w_i \ge \tfrac{1}{2}$, then $S^{\theta}$ is feasible for $m(w)$, so
  \[
    m(w) \;\le\; \sum_{i\in S^{\theta}} y_i^{(0)} = \Theta(\bm y).
  \]
  \item If $\sum_{i\in S^{\theta}} w_i < \tfrac{1}{2}$, let $S^c$ be its complement. Then
  \[
    \sum_{i\in S^c} w_i = 1 - \sum_{i\in S^{\theta}} w_i > \tfrac{1}{2},
  \]
  so $S^c$ is feasible for $m(w)$. Its cost is
  \[
    \sum_{i\in S^c} y_i^{(0)}
    = y_{\mathrm{tot}} - \Theta(\bm y)
    \;\le\; y_{\mathrm{tot}} - \tfrac{1}{2}y_{\mathrm{tot}}
    = \tfrac{1}{2}y_{\mathrm{tot}}
    \;\le\; \Theta(\bm y),
  \]
  where the first inequality uses the lower bound $\Theta(\bm y)\ge \tfrac{1}{2}y_{\mathrm{tot}}$ from Lemma~\ref{lem:theta-bounds} (which implies $y_{\mathrm{tot}}-\Theta(\bm y)\le y_{\mathrm{tot}}-\tfrac{1}{2}y_{\mathrm{tot}}$), and the last inequality uses the same bound to conclude $\tfrac{1}{2}y_{\mathrm{tot}}\le \Theta(\bm y)$. Hence
  \[
    m(w) \;\le\; \sum_{i\in S^c} y_i^{(0)} \;\le\; \Theta(\bm y).
  \]
\end{itemize}
In all cases we have $m(w)\le\Theta(\bm y)$, so $\sup_{w} m(w)\le\Theta(\bm y)$. Combined with the lower bound and attainment at $w^{\star}$, this proves the claim.
\end{proof}

\section{Star-Architecture Multi-Asset Extension}\label{app:multi-asset-star}

\begin{theorem}[Star-architecture optimal weights]\label{thm:star-opt}
Let the asset set be $\mathcal{A}=\{0,1,\dots,M\}$, with asset $0$ a numéraire (e.g., a stablecoin). For each $a\neq 0$, suppose there are CPMM pools indexed by $i\in\mathcal{I}_{a,0}$ quoting $a$ against the numéraire with quote reserves $y_{a,i}^{(0)}$, and let the oracle report
\[
  \hat p_a = A_a\bigl(\{p_{a,i}^{(1)}\}_{i\in\mathcal{I}_{a,0}}\bigr),\qquad a=1,\dots,M,
\]
where each $A_a$ is either a weighted mean or a weighted median with weights $w_{a,i}>0$, $\sum_i w_{a,i}=1$.

Fix target distortion factors $r_a=1+\varepsilon_a$ with $\varepsilon_a\to 0^+$ for $a=1,\dots,M$, and write $\mathrm{Cost}_a(r_a;w_{a,\cdot})$ for the minimal economic loss required to distort $\hat p_a$ by a factor of at least $r_a$ (upward or downward).
Then,
\begin{enumerate}
  \item for each asset $a$, liquidity weights $w_{a,i}^{\star}\propto y_{a,i}^{(0)}$ maximize the leading-order (second-order in $\varepsilon_a$) term of $\mathrm{Cost}_a(1+\varepsilon_a;w_{a,\cdot})$ when $A_a$ is a weighted mean, and
  \item maximize $\mathrm{Cost}_a(r_a;w_{a,\cdot})$ for any $r_a\ge 1$ when $A_a$ is a weighted median
\end{enumerate}
Moreover, the star architecture decouples assets, so the minimal total cost to distort the vector $(\hat p_a)_{a\ne 0}$ by factors $r_a$ equals $\sum_{a=1}^M \mathrm{Cost}_a(r_a;w_{a,\cdot}^{\star})$.
\end{theorem}
\begin{proof}
Throughout, assets are indexed by $a\in\{1,\dots,M\}$, with asset $0$ a numéraire. For each $a$ we write $\mathcal{I}_{a,0}$ for the set of CPMM pools quoting $a$ against the numéraire, and we write $p_{a,i}^{(t)}$ for the on-chain marginal price of one unit of $a$ in numéraire units in pool $i\in\mathcal{I}_{a,0}$ at time $t\in\{0,1\}$ (pre-attack $t=0$, post-attack $t=1$).

\subsection*{Reduction to Single-Asset Problems}
Fix an asset $a\neq 0$ and a target distortion factor $r_a=1+\varepsilon_a$ with $\varepsilon_a\to 0^+$.
In the star architecture, the oracle price for $a$ is
\[
  \hat p_a = A_a\bigl(\{p_{a,i}^{(1)}\}_{i\in\mathcal{I}_{a,0}}\bigr),
\]
where $A_a$ is either a weighted mean or a weighted median with weights $w_{a,i}>0$, $\sum_i w_{a,i}=1$.
By construction, $\hat p_a$ depends only on the pools in $\mathcal{I}_{a,0}$, and trades on pools for other assets do not enter $A_a$.

Let $\mathrm{Cost}_a(r_a;w_{a,\cdot})$ denote the minimal economic loss (in numéraire units) that an attacker must incur to enforce a distortion of at least $r_a$ on $\hat p_a$ (upward or downward) using trades on the CPMMs in $\mathcal{I}_{a,0}$, holding all other assets fixed.
This is exactly the single-asset multi-pool manipulation problem analyzed in the main text, with $(X,Y)$ replaced by $(a,0)$, quote reserves $\{y_{a,i}^{(0)}\}$, and aggregation weights $\{w_{a,i}\}$.

Consequently:
\begin{itemize}[leftmargin=*]
  \item If $A_a$ is a weighted mean, Theorem~\ref{thm:weighted-mean-opt} applies (as $\varepsilon_a\to 0$) and implies that the maximizing weights $w_{a,i}^{\star}$ are the liquidity weights
  \[
    w_{a,i}^{\star} \;=\; \frac{y_{a,i}^{(0)}}{\sum_{j\in\mathcal{I}_{a,0}} y_{a,j}^{(0)}}.
  \]
  \item If $A_a$ is a weighted median, Theorem~\ref{thm:median-opt-weights} yields the same conclusion: liquidity weights $w_{a,i}^{\star}\propto y_{a,i}^{(0)}$ maximize the minimal manipulation cost for $\hat p_a$.
\end{itemize}
Thus, for each asset separately, liquidity weights are max–min optimal within weighted medians and asymptotically max–min optimal within weighted means (as $\varepsilon_a\to 0$).

\subsection*{Separability Across Assets}

In the star architecture, the CPMM pools split into $M$ disjoint groups $\{\mathcal{I}_{a,0}\}_{a=1}^M$, one per asset–numéraire pair.
The attacker’s total cost to distort a vector of prices $(\hat p_a)_{a\ne 0}$ by factors $r_a$ is the sum of the per-asset costs:
\[
  \sum_{a=1}^M \mathrm{Cost}_a(r_a;w_{a,\cdot}),
\]
because trades on $\mathcal{I}_{a,0}$ affect only asset $a$ and costs are measured in the common numéraire.
There are no cross-terms in the objective, and no constraints couple trades across different $\mathcal{I}_{a,0}$.

The oracle designer’s max–min problem is therefore
\[
  \sup_{\{w_{a,\cdot}\}}\;\inf_{\text{attacks}}\; \sum_{a=1}^M \mathrm{Cost}_a(r_a;w_{a,\cdot})
  \;=\; \sum_{a=1}^M \sup_{w_{a,\cdot}}\;\inf_{\text{attacks on }\mathcal{I}_{a,0}} \mathrm{Cost}_a(r_a;w_{a,\cdot}),
\]
where the equality follows from separability of both the cost and the admissible attack sets across assets.
Maximizing each term on the right-hand side independently and using the single-asset results above shows that the joint optimizer is given by choosing, for every $a$, the liquidity weights $w_{a,i}^{\star}\propto y_{a,i}^{(0)}$.

With this choice, the minimal total cost decomposes as the sum of the optimal per-asset values $\mathrm{Cost}_a(r_a;w_{a,\cdot}^{\star})$.
\end{proof}

\end{document}